\documentclass[11pt]{article}

\usepackage{amssymb, amsmath,amsthm}
\usepackage{graphics}

\newtheorem{theorem}{Theorem}
\newtheorem{clm}{Claim}[theorem]

\newtheorem{lemma}{Lemma}

\def \no {\noindent}
\title{Independent Sets in Classes Related to Chair/Fork-free Graphs}

\author{T. Karthick\thanks{Computer Science Unit, Indian Statistical
Institute, Chennai Centre, Chennai-600113, India.
E-mail: {karthick@isichennai.res.in}}}

\begin{document}
\maketitle

\begin{abstract}
The \textsc{Maximum Weight Independent Set (MWIS)} problem on graphs with vertex weights
asks for a set of pairwise nonadjacent vertices of maximum total weight. MWIS is known to
be $NP$-complete in general, even under various restrictions. Let $S_{i,j,k}$ be the graph
consisting of three induced paths of lengths $i, j, k$ with a common initial vertex. The complexity of
the MWIS problem for $S_{1, 2, 2}$-free graphs, and for $S_{1, 1, 3}$-free graphs are open.
In this paper, we show that the MWIS problem can solved in polynomial time
for ($S_{1, 2, 2}$, $S_{1, 1, 3}$, co-chair)-free graphs, by analyzing the structure  of the subclasses of this class of graphs. This  extends some known results  in the literature.
\end{abstract}

\no{\bf Keywords}: Graph algorithms; Independent sets;  Claw-free graphs; Chair-free graphs;  Clique separators; Modular decomposition.

\section{Introduction}
For notation and terminology not defined here, we follow \cite{BLeS}.
Let $P_n$ and $C_n$ denote respectively
the path, and the cycle on $n$ vertices.
If $\cal{F}$ is a family of graphs, a graph $G$ is said to be
{\it $\cal{F}$-free} if it contains no induced subgraph isomorphic to
any graph in $\cal{F}$.

In a graph $G$, an {\it independent (or stable) set} is a subset of
mutually nonadjacent vertices in $G$.  The {\textsc{Maximum
Independent Set (MIS)}} problem asks for an independent set of $G$
with maximum cardinality.  The {\textsc{Maximum Weight Independent Set
(MWIS)}} problem asks for an independent set of total maximum weight
in the given graph $G$ with vertex weight function $w$ on $V(G)$. The M(W)IS
problem is well known to be $NP$-complete in general and hard to
approximate; it remains $NP$-complete even on restricted classes of
graphs \cite{Corneil, Poljak}.  On the other hand,
the M(W)IS problem is known to be solvable in polynomial time on many
graph classes such as: chordal graphs \cite{Frank}; $P_4$-free graphs \cite{CPS}; perfect graphs \cite{GLS}; $2K_2$-free graphs
\cite{Farber}; claw-free graphs \cite{Minty}; fork-free graphs \cite{LM}; apple-free graphs \cite{BLM};  and $P_5$-free graphs \cite{LVV}.

 For
integers $i, j, k \geq 0$, let $S_{i, j, k}$ denote a tree with
exactly three vertices of degree one, being at distance $i$, $j$ and
$k$ from the unique vertex of degree three. The graph $S_{0,1,2}$ is isomorphic to $P_4$ and the graph $S_{0,2,2}$ is isomorphic to $P_5$.  The graph $S_{1,1,1}$ is
called a {\it claw} and $S_{1, 1, 2}$ is called a {\it chair or fork}.
Also, note that $S_{i,j, k}$ is a subdivision of a claw, if $i, j, k \geq 1$.

Alekseev \cite{Alek} showed that the M(W)IS problem remains $NP$-complete on $H$-free graphs, whenever $H$ is connected, but neither a path nor a subdivision of the claw.  As mentioned above, the complexity status of the MWIS problem in the graphs classes defined
by a single forbidden induced subgraph of the form $S_{i,j,k}$ was solved for the case $i + j + k \leq 4$. However, for larger $i + j + k$,
the complexity of MWIS in $S_{i,j,k}$-free graphs is unknown. In particular, the class of
$P_6$-free graphs, the class of $S_{1, 2, 2}$-free graphs, and the
class of $S_{1, 1, 3}$-free graphs constitute the minimal classes,
defined by forbidding a single connected subgraph on six vertices, for
which the computational complexity of M(W)IS problem is unknown. Also, it is known that there is an $n^{O(\log ^2n)}$-time, polynomial-space algorithm for MWIS on $P_6$-free graphs \cite{LPL}. This implies that MWIS on $P_6$-free graphs is not
$NP$-complete, unless all problems in $NP$ can be solved in quasi-polynomial time.
On the other hand, MWIS is shown to be solvable in polynomial time for several subclasses of $S_{i,j,k}$-free graphs, for $i+j+k \geq 5$ such as:  ($P_6$, triangle)-free graphs \cite{BKM}; ($P_6, K_{1,p}$)-free graphs \cite{LR}; ($P_6, C_4$)-free \cite{MSK, Mosca-1}; ($P_6$, diamond)-free graphs \cite{Mosca-2}; ($P_6$, banner)-free graphs \cite{K2}; ($P_6$, co-banner)-free graphs \cite{Mosca-3};  ($P_6$, $S_{1,2,2}$, co-chair)-free graphs \cite{KM-2}; ($S_{1, 1, 3}$, banner)-free graphs \cite{KM}; and ($S_{1, 2, 2}$, bull)-free graphs \cite{KM}.
 It is also known that the MIS problem can
be solved in polynomial time for some subclasses of $S_{i, j, k}$-free
graphs such as: $S_{1,2,k}$-free planar graphs and $S_{1,k,k}$-free graphs of low degree \cite{LM-soda},
and $S_{2,2,2}$-free sub-cubic graphs \cite{LMR}; and see  \cite[Table 1]{GHL} for several other subclasses. See Figure~\ref{sg} for some of the special graphs used in this
paper.

\begin{figure}[t]
\centering
  \includegraphics{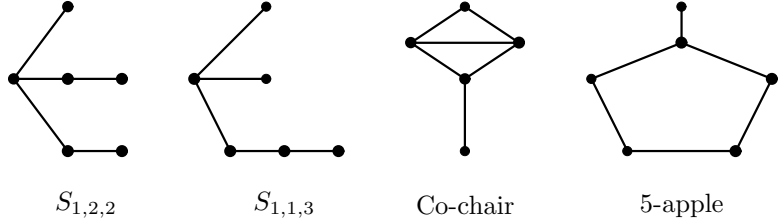}
\caption{Some special graphs.}
\label{sg}
\end{figure}

In this paper,  we show  that the MWIS problem can be efficiently solved in the class of ($S_{1, 2,
2}$, $S_{1, 1,3}$, co-chair)-free graphs by analyzing the structure  of the subclasses of this class of graphs.  This result extends some known results in
the literature such as: the aforementioned results for $P_4$-free graphs, and ($P_5$, co-chair)-free graphs \cite{K1}. A preliminary version (extended abstract)of this paper is appearing in \cite{K3}.

\section{Notation and terminology}\label{NT}

Let $G$ be a finite, undirected and simple graph with vertex-set
$V(G)$ and edge-set $E(G)$.  We let $|V(G)| = n$ and $|E(G)| = m$.  For a
vertex $v \in V(G)$, the {\it neighborhood} $N(v)$ of $v$ is the set
$\{u \in V(G) \mid uv \in E(G)\}$, and the {\it closed neighborhood}
$N[v]$ is the set $N(v) \cup \{v\}$.  The neighborhood $N(X)$ of a
subset $X \subseteq V(G)$ is the set $\{u \in V(G)\setminus X : u$ $
\mbox{~is adjacent to a vertex of }X\}$.  Given a subgraph $H$ of $G$
and $v \in V(G)\setminus V(H)$, let $N_H(v)$ denote the set $N(v) \cap
V(H)$, and for $X \subseteq V(G)\setminus V(H)$, let $N_H(X)$ denote
the set $N(X) \cap V(H)$.  For any two subsets $S$, $T\subseteq V(G)$,
we say that $S$ is \emph{complete} to $T$ if every vertex in $S$ is
adjacent to every vertex in $T$.

 A \emph{hole} is a chordless cycle $C_k$, where $k \geq 5$. An \emph{odd hole} is a hole $C_{2k+1}$, where $k \geq 2$.

 The
\emph{k-apple} is the graph obtained from a chordless cycle $C_k$ of
length $k\geq 4$ by adding a vertex that has exactly one neighbor on
the cycle.

 The \emph{diamond} is the graph $K_4 -e$ with  vertex-set  $\{v_1, v_2, v_3, v_4\}$ and edge-set $\{v_1v_2, v_2v_3, v_3v_4, $ $v_4v_1, v_1v_3\}$.

The \emph{co-chair} is the graph with   vertex-set $\{v_1, v_2, v_3, v_4, v_5\}$ and edge-set
$\{v_1v_2, $ $ v_2v_3, v_3v_4, $ $ v_4v_1, v_1v_3, v_4v_5\}$; it is the complement graph of the  \emph{chair/fork}  graph (see Figure~\ref{sg}).

A vertex $z \in V(G)$ {\it distinguishes} two other vertices $x, y \in
V(G)$ if $z$ is adjacent to one of them and nonadjacent to the other.
A set $M\subseteq V(G)$ is a {\it module} in $G$ if no vertex from
$V(G) \setminus M$ distinguishes two vertices from $M$.  The {\it
trivial modules} in $G$ are $V(G)$, $\emptyset$, and all one-vertex
sets.  A graph $G$ is {\it prime} if it contains only trivial modules.
Note that prime graphs on at least three vertices are connected.

A class of graphs $\cal{G}$ is {\it hereditary} if every induced
subgraph of a member of $\cal{G}$ is also in $\cal{G}$.  We will use
the following theorem by L\"ozin and Milani\v{c} \cite{LM}.

\begin{theorem}[\cite{LM}]\label{thm:LM}
Let $\cal{G}$ be a hereditary class of graphs.  If there is a constant
$p \geq 1$ such that the MWIS problem can be solved in time
$O(|V(G)|^p)$ for every prime graph $G$ in $\cal{G}$, then the MWIS
problem can be solved in time $O(|V(G)|^p + |E(G)|)$ for every graph
$G$ in $\cal{G}$. \hfill $\Box$
\end{theorem}

Let $\cal{C}$ be a class of graphs.  A graph $G$ is {\it nearly
$\cal{C}$} if for every vertex $v$ in $V(G)$ the graph induced by
$V(G) \setminus N[v]$ is in $\cal{C}$.  Let $\alpha_w(G)$ denote the
weighted independence number of $G$.  Obviously, we have:
\begin{eqnarray}
\alpha_w(G) & = & \max\{w(v) + \alpha_w(G\setminus N[v]) \mid v \in
V(G)\}.
\end{eqnarray}
Thus, whenever MWIS is solvable in time $T$ on a class $\cal{C}$, then
it is solvable on nearly $\cal{C}$ graphs in time $n\cdot T$.

A {\it clique} in $G$ is a subset of pairwise adjacent vertices in
$G$.  A {\it clique separator} (or {\it clique cutset}) in a connected
graph $G$ is a subset $Q$ of vertices in $G$ which induces a complete
graph, such that the graph induced by $V(G) \setminus Q$ is disconnected.  A
graph is an {\it atom} if it does not contain a clique separator.

We will also use the following theorem given in \cite{KM}.

\begin{theorem}[\cite{KM}]\label{thm:Tar}
Let $\cal C$ be a class of graphs such that MWIS can be solved in time
$O(f(n))$ for every graph in $\cal C$ with $n$ vertices.  Then in any hereditary
class of graphs whose all atoms are nearly $\cal C$ the MWIS problem
can be solved in time $O(n^2\cdot f(n))$.  \hfill $\Box$
\end{theorem}

The following notation will be used several times in the proofs.
Given a graph $G$, let $v$ be a vertex in $G$ and $H$ be an induced
subgraph of $G\setminus\{v\}$ such that $v$ has no neighbor in $H$.
Let $t=|V(H)|$.  Then we define the following sets:
\begin{eqnarray*}
Q_{\ } &=& \mbox{the component of $G\setminus (V(H)\cup N(V(H)))$ that
contains $v$}, \\
A_i &=& \{x \in V(G) \setminus V(H) \mid |N_H(x)| = i\} \ (1 \leq i
\leq t), \\
A_i^+ &=& \{x \in A_i \mid N(x) \cap Q \neq \emptyset\}, \\
A^-_i &=& \{x \in A_i \mid N(x) \cap Q = \emptyset\}, \\
A^+ &=&   A^+_1\cup\cdots\cup A^+_t \ \mbox{and} \
A^- =  A^-_1\cup\cdots\cup A^-_t.
\end{eqnarray*}
So, $N(H) = A^+ \cup A^-$.  Note that, by the definition of $Q$ and
$A^+$, we have $A^+ = N(Q)$.  Hence $A^+$ is a separator between $H$
and $Q$ in $G$.

\section{Preliminary lemmas}
\begin{figure}
\centering
    \includegraphics{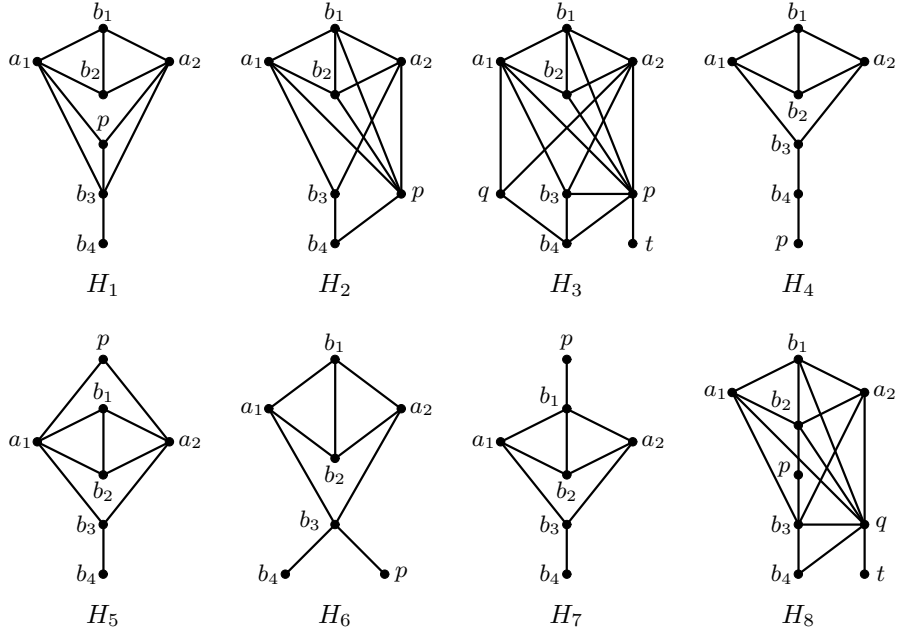}
\caption{Graphs $H_i$,  $i \in \{1, 2, \ldots, 8\}$ used in Lemma~\ref{primeSSCfree-forb}.}
\label{figure-lemmas}
\end{figure}

\begin{lemma}\label{primeSSCfree-forb}
Let $G = (V, E)$ be a prime  co-chair-free graph. Then $G$ is $(H_1, H_2, H_3)$-free.
Further, if $G$ is $S_{1, 2,2}$-free, then $G$ is $(H_4, H_5)$-free, and if $G$ is $S_{1, 1, 3}$-free, then $G$ is $(H_4, H_6, $ $H_7, H_8)$-free.
See Figure~\ref{figure-lemmas} for the graphs $H_i$,  $i \in \{1, 2, \ldots, 8\}$.
\end{lemma}
\begin{proof} Suppose to the contrary that $G$ contains an induced $H_i$, for some $i \in \{1, 2, \ldots, 8\}$ (as shown in Figure~\ref{figure-lemmas}). Since $G$ is prime, $\{a_1, a_2\}$ is not a
module in $G$, so there exists a vertex $x \in V \setminus
V(H_i)$ such that (up to symmetry) $xa_1 \in E$ and $xa_2 \notin
E(G)$.  Then since $\{x, a_1, b_1, b_2, a_2\}$ does not induce a co-chair in $G$,  $x$ is adjacent to one of $b_1, b_2$.

Suppose that $x$ is adjacent to both of $b_1, b_2$.  Then since $\{x, b_1, b_2, a_2, b_3\}$ does not induce a co-chair in $G$, $xb_3 \in E$, and since
 $\{x, b_1, b_2, a_2, b_4\}$ does not induce a co-chair in $G$, $xb_4 \notin E$. But, now $\{x, a_1, b_1, b_3, b_4\}$ induces a co-chair in $G$, which is a contradiction.  Therefore, $x$ is adjacent to exactly one of $b_1, b_2$.

 Suppose that $i \neq 7, 8$. We may assume (up to symmetry) that $xb_1 \in E$ and
$xb_2 \notin E$.  Since $\{x, a_1,  b_1, b_2, b_4\}$ does not induce a co-chair in $G$, $xb_4 \notin E$, and then since $\{x, a_1, b_1, b_3, b_4\}$ does not induce a co-chair in $G$, $xb_3 \notin E$. Now, we prove a contradiction as follows:

\no{$i = 1$}: Since $\{x, a_1, a_2, b_3, p\}$ does not induce a co-chair in $G$, $xp \in E$. But, now $\{x, a_1, b_3, b_4, p\}$ induces a co-chair in $G$, which is a contradiction.  Thus, $G$ is $H_1$-free.

 \no{$i = 2$}: Since $\{x, a_1, b_1, p, b_4\}$ does not induce a co-chair in $G$, $xp \in E$. But, now $\{x, b_1, p, a_2, b_3\}$ induces a co-chair in $G$, which is a contradiction.  Thus, $G$ is $H_2$-free.

\no{$i = 3$}:
 Since $\{x, a_1,  b_1, b_2,t\}$ does not induce a co-chair in $G$, $xt \notin E$, and then since $\{x, a_1,  b_1,  p, t\}$ does not induce a co-chair in $G$, $xp \in E$. Then since $\{x, b_1, p, a_2, q\}$ does not induce a co-chair in $G$, $xq \in E$. But, now $\{b_1, x, a_1, q, b_4\}$ induces a co-chair in $G$, which is a contradiction.  Thus, $G$ is $H_3$-free.

 \no{$i = 4$}: Since  $\{x, a_1,  b_1, b_2, p\}$ does not induce a co-chair in $G$,  $xp \notin E$. But, now $\{x, a_1, b_3, b_4, p,$ $ a_2\}$ induces an
 $S_{1, 2,2}$ in $G$ or $\{x, a_1, b_3, b_4, p, b_2\}$ induces an $S_{1, 1, 3}$ in $G$, a contradiction.  Thus, $G$ is $H_4$-free.

 \no{$i = 5$}:  Since $\{x, b_1, a_2, b_3, b_4, p\}$ does not induce an $S_{1, 2, 2}$ in $G$, $xp \in E$. But, now $\{x, p, a_2, b_3,$ $ b_4, b_2\}$ induces an $S_{1, 2, 2}$ in $G$, which is a contradiction. Thus, $G$ is $H_5$-free.

\no{$i = 6$}: Since $\{x, a_1,  b_1, b_2, p\}$ does not induce a co-chair in $G$, $xp \notin E$. But, now $\{x, b_1, a_2, b_3,$ $ b_4, p\}$ induces an $S_{1, 1, 3}$ in $G$, which is a contradiction.  Thus, $G$ is $H_6$-free.

Suppose that $i = 7$. Note that $x$ is adjacent to exactly one of $b_1, b_2$. Then as earlier $xb_3, xb_4 \notin E$ (otherwise, $G$ induces a co-chair).  Now, if $xb_1 \in E$ and $xb_2 \notin E$, then since $\{x, b_1, a_2, b_3, b_4, p\}$ does not induce an $S_{1, 1, 3}$ in $G$, $xp \in E$. But, now $\{p, x, b_1, a_1, b_3\}$ induces a co-chair in $G$, which is a contradiction. Next, if $xb_2 \in E$ and $xb_1 \notin E$,  then since $\{x, a_1,  b_2, b_1, p\}$ does not induce a co-chair in $G$, $xp \in E$.  But, now  $\{b_4, b_3, a_1, x, p, a_2\}$ induces an $S_{1, 1, 3}$ in $G$, which is a contradiction.
Thus, $G$ is $H_7$-free.

Suppose that $i = 8$.  Again as earlier $xb_3, xb_4 \notin E$ (otherwise, $G$ induces a co-chair).  Now, if $xb_1 \in E$ and $xb_2 \notin E$, then since $\{x, a_1,  b_1, b_2, p\}$ does not induce a co-chair in $G$, $xp \in E$. Also, since $\{x, a_1, b_1, b_2, t\}$ does not induce a co-chair in $G$, $xt \notin E$, and then since $\{x, a_1, b_1, q, t\}$ does not induce a co-chair in $G$,
 $xq \in E$. But, now $\{a_2, b_1, q, x, p\}$ induces a co-chair in $G$, which is a contradiction.
Next, if $xb_2 \in E$ and $xb_1 \notin E$, then since $\{x, b_2, p, b_3, b_4, b_1\}$ does not induce an $S_{1, 1, 3}$ in $G$, $xp \in E$. Also, since $\{x, a_1, b_1, b_2, t\}$ does not induce a co-chair in $G$, $xt \notin E$, and then since $\{x, a_1,  b_2, q, t\}$ does not induce a co-chair in $G$, $xq \in E$. But, now $\{p, x, b_2, q, t\}$ induces a co-chair in $G$, which is a contradiction.
Thus, $G$ is $H_8$-free.

This completes the proof of Lemma~\ref{primeSSCfree-forb}.
\end{proof}

\begin{figure}[t]
\centering
  \includegraphics{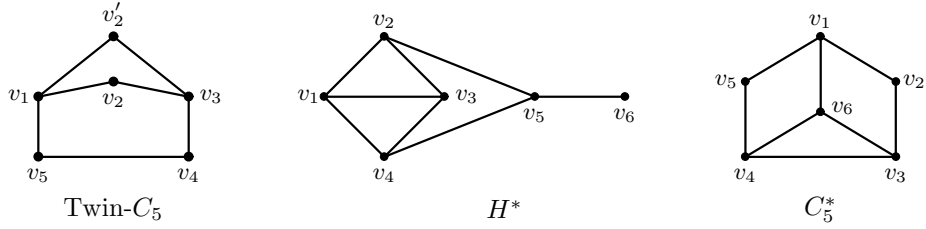}
\caption{Graphs twin-$C_5$,  $H^*$ and $C_5^*$.}
\label{HC}
\end{figure}

\begin{lemma} \label{5apple-C5*-diamond-free-implies-twin-C5-free}
If $G = (V, E)$ is a prime ($5$-apple, $C_5^*$, diamond)-free graph, then $G$ is twin-$C_5$-free.
\end{lemma}
\begin{proof}
Suppose to the contrary that $G$ contains an induced twin-$C_5$, say $H$ as shown in
Figure~\ref{HC}.  Since $G$ is prime, $\{v_2, v_2'\}$ is not a
module in $G$, so there exists a vertex $x$ in $V \setminus
V(H)$ such that (up to symmetry) $xv_2' \in E$ and $xv_2 \notin
E$.  Then since $\{v_2', v_1, v_3, v_4, v_5, x\}$ does not induce a 5-apple in $G$, $xv_i \in E$, for some $i \in \{1, 3, 4, 5\}$.
If $xv_1 \in E$, then since $G$ is diamond-free, $xv_3, xv_5 \notin E$. Then since $\{v_1, v_2, v_3, v_4, v_5, x\}$ does not induce a 5-apple in $G$, $xv_4 \in E$, but then $\{v_1, v_2', v_3, v_4, v_5, x\}$ induces a $C_5^*$ in $G$, which is a contradiction. A similar contradiction arises if we assume $xv_3 \in E$. So, we may assume that $xv_1, xv_3 \notin E$. Then since $G$ is 5-apple-free, $xv_4\in E$ and $xv_5 \in E$. Now, $\{v_1, v_2', v_3, v_4, v_5, x\}$ induces a $C_5^*$ in $G$, which is a contradiction. So, $G$ is twin-$C_5$-free.
\end{proof}

\begin{lemma}[\cite{K1}]\label{GC-freeimpliesD-free}
If $G = (V, E)$ is a prime  (co-chair, gem)-free graph, then $G$ is diamond-free.
\end{lemma}


\section{($S_{1, 2, 2}$, $S_{1, 1, 3}$, diamond)-free graphs}

In this section, we show  that the MWIS problem can be efficiently solved in the class of ($S_{1, 2,
2}$, $S_{1, 1,3}$, diamond)-free graphs by analyzing the atomic structure of the subclasses of this class of graphs.

\subsection{($S_{1, 2, 2}$, $S_{1, 1, 3}$, diamond, 5-apple, $C_5^*$)-free graphs}
\begin{theorem} \label{SSDC5-free-odd-hole impliesclaw-free}
Let $G= (V, E)$ be a prime ($S_{1, 2, 2}$, $S_{1, 1, 3}$,  diamond, $5$-apple, $C_5^*$)-free graph. If $G$ contains an  odd hole $C_{2k+1}$ with $k \geq 2$, then $G$ is claw-free.
\end{theorem}

\begin{proof}
 Since $G$ is prime, it is connected, and by Lemma~\ref{5apple-C5*-diamond-free-implies-twin-C5-free}, $G$ is twin-$C_5$-free.
  Let $C$ denotes a shortest odd hole $C_{2k+1}$  in $G$ with vertices $v_1, v_2, \ldots,  v_{2k+1}$ and edges $v_iv_{i+1}, v_{2k+1}v_1 \in E$, where $i \in \{1, 2, \ldots, 2k\}$ with $k \geq 2$. Then it is verified that the  following claim holds.

  \begin{clm}\label{nei-odd-hole} If  $x \in V(G) \setminus V(C)$ has a neighbor on $C$, then  there exists an $i$ such that $N(x) \cap V(C) = \{v_i, v_{i+1}\}$.
\end{clm}

\no{\bf Proof of Claim~\ref{nei-odd-hole}}:  If $k = 2$, since $G$ is ($5$-apple, $C_5^*$, twin-$C_5$, diamond)-free, the claim holds. So, suppose that $k \geq 3$. To prove the claim, we prove the following:

\no{\bf (1)} There exists an $i$ such that $xv_i, xv_{i+1} \in E$ and $xv_{i-1}, xv_{i+2} \notin E$.

 \no{\bf (2)}  Either $|N(x) \cap V(C)| = 2$ or $|N(x) \cap V(C)| = 4$. Moreover, there exists an $i$ such that $N(x) \cap V(C) = \{v_i, v_{i+1}\}$ (if $|N(x) \cap V(C)| = 2$), and $N(x) \cap V(C) = \{v_i, v_{i+1}, v_j, v_{j+1}\}$, for some $j \in \{i+3, i+4, \ldots, i+2k-2\}$ (if $|N(x) \cap V(C)| = 4$).

Since $x$ has a neighbor on $C$, we may assume that $x$ is adjacent to $v_{i}$ on $C$.
 If (1) does not hold, then $xv_{i+1}, xv_{i-1} \notin E$.  Then since $\{v_{i-2}, v_{i-1}, v_{i}, v_{i+1}, v_{i+2}, x\}$ does not induce an $S_{1, 2, 2}$ in $G$, we have either $xv_{i-2} \in E$ or $xv_{i+2} \in E$. We may assume, up to symmetry, that $xv_{i-2} \in E$. Then since $\{v_{i+3}, v_{i+2}, v_{i+1}, v_i, v_{i-1}, x\}$ does not induce a $C_5$ or an $S_{1, 1, 3}$ in $G$, we have $xv_{i+2} \in E$.
   Then since $\{v_{i+1}, v_{i+2}, x, v_{i-2}, $ $v_{i-1}, v_{i-3}\}$ does not induce an $S_{1, 1, 3}$ in $G$,  $xv_{i-3} \in E$. Then since $G$ is diamond-free, $\{v_{i+3}, v_{i-3},$ $ x, v_i, v_{i+1}, v_{i-1}\}$ induces an $S_{1, 1, 3}$ in $G$ (if $k = 3$) or $\{v_{i-4}, v_{i-3}, x, v_i, v_{i+1}, v_{i-1}\}$ induces an $S_{1, 1, 3}$ in $G$ (if $k \geq 4$), a contradiction. So (1) holds.

 By (1), we have $\{v_i, v_{i+1}\} \subseteq N(x) \cap V(C)$, and $xv_{i-1}, xv_{i+2} \notin E$.  Further, if there exists an index $j \in \{i+3, i+4, \ldots, i+2k-1\}$ such that $xv_j \in E$ and $xv_{j-1} \notin E$, then $xv_{j+1} \in E$ (for, otherwise, $\{v_{i-1}, v_i, v_{i+1}, v_{i+2}, x\} \cup \{v_{j-1}, v_j, v_{j+1}\}$ induces an $S_{1, 1, 3}$ in $G$).  Now, if $x$ is adjacent to a vertex $v_t$ on $C$, where $t \notin \{i-1, i, i+1, i+2, j-1, j, j+1, j+2\}$, then either  a diamond or an $S_{1, 2, 2}$ is an induced subgraph of $G$, which is a contradiction. Hence $N(x) \cap V(C) = \{v_i, v_{i+1}\}$ or $N(x) \cap V(C) = \{v_i, v_{i+1}, v_j, v_{j+1}\}$, for some $j \in \{i+3, i+4, \ldots, i+2k-2\}$. So (2) holds.

Further, if $|N(x) \cap V(C)| = 4$, then $G$ contains an odd hole $C'$ shorter than $C$, which is a contradiction to the choice of $C$. $\Diamond$

To prove the theorem, we suppose for contradiction that $G$ contains an induced claw, say $K$ with vertex-set $\{a,  b, c, d\}$ and edge-set $\{ab, ac, ad\}$. By Claim~\ref{nei-odd-hole},
 $K$ cannot have more than two vertices on $C$. Also, at most one vertex in $\{b, c, d\}$ belongs to $C$. Now we have following cases (the other cases are symmetric):
\begin{enumerate}
\item[(1)] $V(K) \cap V(C) = \{a, d\}$: Let $y$ be the other neighbor of $a$ on $C$. Then by Claim~\ref{nei-odd-hole}, $by, cy \in E$. But, now $\{a, b, c, y\}$
induces a diamond in $G$, which is a contradiction.
\item[(2)]  $V(K) \cap V(C) = \{a\}$: We may assume (wlog.) that  $a = v_1$. Then by Claim~\ref{nei-odd-hole}, at least two vertices in $\{b, c, d\}$
are  adjacent either to  $v_2$ or to $v_{2k+1}$, say $b$ and $c$ are adjacent to $v_2$. Then $\{a, v_2, b, c\}$ induces a diamond in $G$, which is a contradiction.
\item[(3)] $V(K) \cap V(C) = \{d\}$: We may assume (wlog.) that  $d = v_1$. Then by Claim~\ref{nei-odd-hole}, up to symmetry, we may assume $av_2 \in E$.
Suppose that $k =2$. Then since $G$ is ($S_{1, 1, 3}$, $5$-apple, diamond)-free, both $b$ and $c$ have neighbors in $C$. To avoid a diamond in $G$ and by Claim~\ref{nei-odd-hole}, we assume (wlog.) that $bv_3, bv_4, cv_4, cv_5 \in E$. But, now $\{d=v_1, a, b, v_4, v_5, c\}$ induces a $C_5^*$ in $G$, which is a contradiction.
So, suppose that $k \geq 3$. Then since $G$ is diamond-free, $bv_2, cv_2 \notin E$. We claim that either $bv_3 \in E$ or $cv_3 \in E$. Otherwise, since $\{v_4, v_3, v_2, a, b, c\}$ does not induce an $S_{1, 1, 3}$ in $G$, either $bv_4 \in E$ or $cv_4 \in E$. But, now $\{v_4, v_3, v_2, a, b, c\}$ induces a $C_5$ in $G$, which is a contradiction to the fact that $k \geq 3$ and the choice of $C$. Thus, we may assume that $bv_3 \in E$. Then by Claim~\ref{nei-odd-hole}, $bv_4  \in E$. Then $cv_3 \notin E$ (for, otherwise, by Claim~\ref{nei-odd-hole}, $cv_4 \in E$, but then $\{v_3, v_4, b, c\}$ induces a diamond in $G$), and hence $cv_4 \notin E$ (for, otherwise, $\{a, c, v_4, v_3, v_2\}$ induces a $C_5$ in $G$). Now, $\{v_5, v_4, b, a, c\}$ induces a $C_5$ in $G$ (if $cv_5 \in E$) or $\{v_5, v_4, b, a, d(= v_1), c\}$ induces an $S_{1, 1, 3}$ in $G$ (if $cv_5 \notin E$), a contradiction.

\item[(4)] $V(K) \cap V(C) = \emptyset$ and a vertex of $K$ has a neighbor on $C$:  Assume $a$ has neighbors on $C$, say $v_1$ and $v_2$. Then to avoid an induced claw intersecting $C$, both $v_1$ and $v_2$ have exactly two neighbors among $b, c, d$. We may assume (wlog.) that $v_1$ is adjacent to $b$ and $c$.   But, now $\{v_1, a, b, c\}$ induces a diamond in $G$, which is a contradiction. So, assume that $a$ has no neighbor on $C$. Assume (wlog.) that $b$ has a neighbor on $C$. By Claim~\ref{nei-odd-hole}, we may assume that $N(b) \cap V(C) = \{v_1, v_2\}$. Then since $G$ is $S_{1, 1, 3}$-free, both $c$ and $d$ have neighbors on $C$. Now, $v_2$ is adjacent to either $c$ or $d$ (for, otherwise, $\{v_3, v_2, b, a, c, d\}$ will induce an $S_{1, 1, 3}$ or a $C_5$ in $G$). Assume that $cv_2 \in E$. Since $G$ is diamond-free, $cv_1 \notin E$ and by Claim 1, $cv_3 \in E$. Thus, $N(c) \cap V(C) = \{v_2, v_{3}\}$. By similar arguments, we see that $N(d) \cap V(C) = \{v_1, v_{2k+1}\}$. Now, $\{v_4, v_3, c, a, b ,d\}$ induces an $S_{1, 1, 3}$ in $G$, which is a contradiction.

\item[(5)] $V(K)\cap V(C) = \emptyset$ and no vertex of $K$ has a neighbor on $C$: Since $G$ is connected, there exists an $i \in \{1, 2, \ldots, 2k+1\}$ and a path
$v_i = u_1-u_2-\cdots-u_t-a$, say $P$ connecting $v_i$ and $a$ in $G$ (where $t \geq 2$) and with $u_2$ has a neighbor on $C$. By the choice of $P$, no vertex of this path has a neighbor on $C$ except $u_2$. By Claim~\ref{nei-odd-hole}, either $u_2v_{i+1} \notin E$ or $u_2v_{i-1} \notin E$. Assume that $u_2v_{i+1} \notin E$. Now, $u_t \neq b, c, d$ (for, otherwise (wlog.) if $u_t = b$, then $\{v_{i+1}, v_i = u_1, u_2, \ldots, u_t = b, a, c, d\}$ induces an $S_{1, 1, 3}$ in $G$). Then since $G$ is diamond-free, at least two vertices in $\{b, c, d\}$ are not adjacent to $u_t$, say $b$ and $c$. Then $\{v_{i+1}, v_i = u_1, u_2, \ldots, u_t, a, b, c\}$ induces an $S_{1, 1, 3}$ in $G$, which is a contradiction.
\end{enumerate}

\noindent{Hence} $G$ is claw-free, and this completes the proof of the theorem.
\end{proof}

\begin{theorem} \label{mwis-SSD5-appleC5*-free-time}
The MWIS problem can be solved in polynomial time for ($S_{1, 2, 2}$, $S_{1, 1, 3}$, diamond, 5-apple, $C_5^*$)-free graphs.
\end{theorem}

\begin{proof} Let $G$ be an ($S_{1, 2, 2}$, $S_{1, 1, 3}$,  diamond, $5$-apple, $C_5^*$)-free
graph.
If $G$ is odd-hole-free, then $G$ is (odd-hole, diamond)-free. Since MWIS in (odd-hole, diamond)-free graphs can be solved in polynomial time \cite{BM-odd-hole}, MWIS can be solved in polynomial time for $G$. Suppose that $G$ is prime
  and contains an odd-hole. Then by Theorem~\ref{SSDC5-free-odd-hole impliesclaw-free}, $G$ is claw-free. Since  MWIS in claw-free graphs can be solved in polynomial time \cite{Minty}, MWIS can be solved in polynomial time for $G$.  Then the time
complexity is the same when $G$ is not prime, by Theorem
\ref{thm:LM}.
\end{proof}

\subsection{($S_{1, 2, 2}$, $S_{1, 1, 3}$, diamond, 5-apple)-free graphs}

\begin{theorem} \label{SSD5-apple-freeimpliesc5*-free}
Let $G = (V, E)$ be an ($S_{1, 2, 2}$, $S_{1, 1, 3}$, diamond, 5-apple)-free graph. Then  $G$ is nearly $C_5^*$-free.
\end{theorem}

\begin{proof} Let us assume on the contrary
that there is a vertex $v \in V(G)$ such that $G\setminus N[v]$
contains an induced $C_5^*$, say $H$, with vertices named as in Figure~\ref{HC}.
Let $C$ denotes the $5$-cycle induced by the vertices $\{v_1, v_2,
v_3, v_4, v_5\}$ in $H$. For $i \in \{1, 2, \ldots, 6\}$, we define sets $A_i$,
$A_i^+$, $A^+$,  and $Q$ as in the last paragraph of
Section~\ref{NT}. To prove the theorem, it is enough to show that $A^+ = \emptyset$.
Assume to the contrary that $A^+ \neq \emptyset$, and let $x \in A^+$. Then there exists a vertex $z \in Q$ such that
$xz \in E$. Then since $G$ is (5-apple, diamond)-free, $|N_H(x) \cap V(C)| \in \{0, 2, 3\}$. Now:
\begin{enumerate}
\item[(i)] If $|N_H(x) \cap V(C)|= 0$, then since $x\in N(H)$, $xv_6 \in E$. But then $\{z, x, v_6, v_1,$ $ v_2, v_5\}$ induces an $S_{1, 1, 3}$ in $G$, which is a contradiction.

\item[(ii)] If $|N_H(x) \cap V(C)|= 2$, and if $N_H(x) \cap V(C) = \{v_i, v_{i+2}\}$, for some $i \in \{1, 2, 3, 4, 5\}$, $i \mod 5$, then
$\{z, x, v_{i+2}, v_{i+3}, v_{i+4}, v_{i}\}$ induces a 5-apple in $G$, which is a contradiction.

\item[(iii)] If $|N_H(x) \cap V(C)|= 2$, and if $N_H(x) \cap V(C) = \{v_i, v_{i+1}\}$, for some $i \in \{1, 2, 3, 4, 5\}$, $i \mod 5$, then since $\{z, x\} \cup V(H)$ does not induce a diamond or an $S_{1, 1, 3}$ in $G$, we have $i \neq 3$. Again, since $G$ is diamond-free,  $xv_6 \notin E$. But, then $\{z, x\} \cup V(H)$ induces either an $S_{1, 1, 3}$ or an $S_{1, 2, 2}$ in $G$, which is a contradiction.

\item[(iv)] If $|N_H(x) \cap V(C)|= 3$,  then since $G$ is diamond-free, $N_H(x) \cap V(C) = \{v_i, v_{i+1}, v_{i+3}\}$, for some $i \in \{1, 2, 3, 4, 5\}$, $i \mod 5$. Then since $G$ is diamond-free, $i \neq 3$ and  $xv_6 \notin E$. But, then $\{z, x\} \cup V(H)$ induces either an $S_{1, 1, 3}$ or an $S_{1, 2, 2}$ in $G$, which is a contradiction.
\end{enumerate}

\noindent{These} contradictions show that $A^+ = \emptyset$, and hence $G$ is nearly $C_5^*$-free.
\end{proof}

\begin{theorem} \label{mwis-SSD5-apple-free-time}
The MWIS problem can be solved in polynomial time  for ($S_{1, 2, 2}$, $S_{1, 1, 3}$, diamond, 5-apple)-free graphs.
\end{theorem}

\begin{proof} Let $G$ be an ($S_{1, 2, 2}$, $S_{1, 1, 3}$, diamond, 5-apple)-free
graph.  Then by Theorem~\ref{SSD5-apple-freeimpliesc5*-free},  $G$ is nearly $C_5^*$-free.  Since  MWIS in  ($S_{1, 2, 2}$, $S_{1, 1, 3}$, diamond, 5-apple, $C_5^*$)-free graphs can be solved in polynomial time (by Theorem~\ref{mwis-SSD5-appleC5*-free-time}), by the consequence given below equation (1) in Section~\ref{NT}, MWIS in  ($S_{1, 2, 2}$, $S_{1, 1, 3}$, diamond, 5-apple)-free graphs can be solved in polynomial time.
\end{proof}

\subsection{($S_{1, 2, 2}$, $S_{1, 1, 3}$, diamond)-free graphs}

\begin{theorem} \label{SSD-freeimplies5apple-free}
Let $G = (V, E)$ be an ($S_{1, 2, 2}$, $S_{1, 1, 3}$, diamond)-free graph. Then every
atom of $G$ is nearly 5-apple-free.
\end{theorem}

\begin{proof} Let $G'$ be an atom of $G$.  We want to show
that $G'$ is nearly 5-apple-free, so let us assume on the contrary
that there is a vertex $v \in V(G')$ such that $G'\setminus N[v]$
contains an induced 5-apple $H$.  Let $H$ have vertex set $\{v_1, v_2,
v_3, v_4, v_5, v_6\}$ and edge set $\{v_1v_2, v_2v_3, v_3v_4, $ $
v_4v_5, v_5v_1, v_1v_6\}$. Let $C$ denotes the $5$-cycle induced by the vertices $\{v_1,  v_2,
v_3, v_4, v_5\}$ in $H$. For $i\in \{1, 2, \ldots, 6\}$, we define sets $A_i$,
$A_i^+$,  $A^+$,  and $Q$, with respect to $G$, $v$ and $H$, as in the last paragraph of
Section~\ref{NT}.  Then, we immediately have the following:

\begin{clm}\label{cl:T3-nei}
 If $x \in N(H)$, then $|N_H(x) \cap V(C)| \leq 3$. In particular, if $x \in A^+$, then $|N_H(x) \cap V(C)| = 3$, and hence there exists an index
$j \in \{1, 2, \ldots, 5\}$, $j \mod 5$, such that $N_H(x) \cap V(C) = \{v_j, v_{j+1}, v_{j+3}\}$. $\Diamond$
\end{clm}

So, we have:
\begin{clm}\label{cl:T3-a1p-a2p-a5p-a6p}
$A_1^+ = A_2^+ = A_5^+ = A_6^+ = \emptyset$.
\end{clm}

\begin{clm}\label{cl:T3-a3p}
If $x \in A_3^+$, then $N_H(x)$ is equal to $\{v_1, v_3, v_4\}$.
\end{clm}

\no{\it Proof of Claim~\ref{cl:T3-a3p}}.  Suppose not.  Then by Claim~\ref{cl:T3-nei}, there exists an index
$j \in \{1, 2, 4,  5\}$,  $j \mod 5$, such that $N_H(x) = \{v_j, v_{j+1}, v_{j+3}\}$. Since $x \in A_3^+$, $xv_6 \notin E$, and there exists a vertex $z$ in $Q$ such that $xz \in E$.  Now, if $N_H(x) = \{v_1, v_2, v_4\}$, then $\{v_6, v_1, x, v_4, v_3, z\}$ induces an
$S_{1, 2, 2}$ in $G$, and if $N_H(x) =\{v_2, v_3, v_5\}$, then $\{v_6, v_1, v_5, x, v_3, z\}$ induces an $S_{1, 1, 3}$ in $G$, a
contradiction.  Since the other cases are symmetric, the claim
follows.  $\Diamond$

\begin{clm}\label{cl:T3-a3pc}
$|A_3^+| = 1$.
\end{clm}

\no{\it Proof of Claim~\ref{cl:T3-a3pc}}.
 Suppose not.  Let  $x, y \in
A_3^+$. By Claim~\ref{cl:T3-a3p}, $N_H(x) = N_H(y) = \{v_1, v_3, v_4\}$. Now, if $xy \in E$, then $\{v_4, x, y, v_1\}$ induces a diamond in $G$,
 and if $xy \notin E$, then $\{v_4, x, y, v_3\}$ induces a diamond in $G$, a contradiction.
  $\Diamond$

\begin{clm}\label{cl:T3-a4p}
If $x \in A_4^+$, then $N_H(x)$ is equal to $\{v_1, v_3, v_4, v_6\}$.
\end{clm}

\no{\it Proof of Claim~\ref{cl:T3-a4p}}.   Suppose not.  Then by Claim~\ref{cl:T3-nei}, there exists an index
$j \in \{1, 2, 4,  5\}$,  $j \mod 5$, such that $N_H(x) \cap V(C) = \{v_j, v_{j+1}, v_{j+3}\}$. Since $x \in A_4^+$, $xv_6 \in E$ and there exists a vertex $z$ in $Q$ such that $xz \in E$.  Now, if $N_H(x) \cap V(C) = \{v_1, v_2, v_4\}$, then $\{v_6, v_1, v_2, x\}$ induces a diamond in $G$, and if $N_H(x) \cap V(C) =\{v_2, v_3, v_5\}$, then $\{v_1, v_6, x, v_3, v_4, z\}$ induces an $S_{1, 2, 2}$ in $G$, a contradiction.  Since the other cases are symmetric, the claim follows. $\Diamond$

\begin{clm}\label{cl:T3-a4pc}
$|A_4^+| = 1$.
\end{clm}

\no{\it Proof of Claim~\ref{cl:T3-a4pc}}.  Suppose not.  Let  $x, y \in
A_4^+$. By Claim~\ref{cl:T3-a4p}, $N_H(x) = N_H(y) = \{v_1, v_3, v_4, v_6\}$. Now, if $xy \in E$, then $\{v_4, x, y, v_3\}$ induces a diamond in $G$, a contradiction
 and if $xy \notin E$, then $\{v_3, v_4, x, y\}$ induces a diamond in $G$, a contradiction.   $\Diamond$

\begin{clm}\label{cl:T3-a3p-a4p}
At most one of $A_3^+$ or $A_4^+$ is non-empty.
\end{clm}

\no{\it Proof of Claim~\ref{cl:T3-a3p-a4p}}.
 Suppose not.  Let  $x \in A_3^+$ and $y \in A_4^+$. By Claim~\ref{cl:T3-a3p}, $N_H(x) = \{v_1, v_3, v_4\}$, and by Claim~\ref{cl:T3-a4p}, $N_H(y) = \{v_1, v_3, v_4, v_6\}$. Now, if $xy \in E$, then $\{v_3, x, y, v_1\}$ induces a diamond in $G$,  a contradiction,  and if $xy \notin E$, then $\{v_3, v_4, x, y\}$ induces a diamond in $G$, a contradiction.
  $\Diamond$

\medskip
Now, by Claim~\ref{cl:T3-a1p-a2p-a5p-a6p}, $A^+ = A_3^+ \cup A_4^+$, and by Claims~\ref{cl:T3-a3pc}, \ref{cl:T3-a4pc} and \ref{cl:T3-a3p-a4p}, $A^+$ is a clique. Since $A^+$ is a separator between $H$ and $Q$ in $G$, we obtain that
$V(G')\cap A^+$ is a clique separator in $G'$ between $H$ and
$V(G')\cap Q$ (which contains $v$).  This is a contradiction to the
fact that $G'$ is an atom.  \end{proof}

\begin{theorem} \label{mwis-SSD-free-time}
The MWIS problem can be solved in polynomial time for ($S_{1, 2, 2}$, $S_{1, 1, 3}$, diamond)-free graphs.
\end{theorem}

\begin{proof} Let $G$ be an ($S_{1, 2, 2}$, $S_{1, 1, 3}$, diamond)-free
graph.  Then by Theorem~\ref{SSD-freeimplies5apple-free},  every atom of $G$ is nearly 5-apple-free.  Since  MWIS in  ($S_{1, 2, 2}$, $S_{1, 1, 3}$, diamond, 5-apple)-free graphs can be solved in polynomial time (by Theorem~\ref{mwis-SSD5-apple-free-time}), MWIS in  ($S_{1, 2, 2}$, $S_{1, 1, 3}$, diamond)-free graphs can be solved in polynomial time, by Theorem~\ref{thm:Tar}.
\end{proof}

\section{($S_{1, 2, 2}$, $S_{1, 1, 3}$, co-chair)-free graphs}

In this section, we show  that the MWIS problem can be efficiently solved in the class of ($S_{1, 2,
2}$, $S_{1, 1,3}$, co-chair)-free graphs by analyzing the atomic structure of the subclasses of this class of graphs.

\subsection{($S_{1, 2, 2}$, $S_{1, 1, 3}$,  co-chair, gem)-free graphs}
\begin{theorem} \label{mwis-SSCG-free-time}
The MWIS problem can be solved in polynomial time for ($S_{1, 2, 2}$, $S_{1, 1, 3}$,  co-chair, gem)-free graphs.
\end{theorem}

\begin{proof} Let $G$ be an ($S_{1, 2, 2}$, $S_{1, 1, 3}$, co-chair, gem)-free
graph.  First suppose that $G$ is prime. Then by
 Lemma~\ref{GC-freeimpliesD-free}, $G$
is diamond-free. Since the MWIS problem in ($S_{1, 2, 2}$, $S_{1, 1, 3}$, diamond)-free graphs can be solved in  polynomial time (by Theorem~\ref{mwis-SSD-free-time}), MWIS can be solved in
polynomial time for $G$, by Theorem~\ref{thm:LM}.  Then the time
complexity is the same when $G$ is not prime, by Theorem
\ref{thm:LM}.
\end{proof}

\subsection{($S_{1, 2, 2}$, $S_{1, 1, 3}$,  co-chair, $H^*$)-free graphs}
\begin{theorem}\label{SSCH*-free-implies-Gem-free}
Let $G = (V, E)$ be a prime ($S_{1, 2, 2}$, $S_{1, 1, 3}$,  co-chair, $H^*$)-free graph.  Then every
atom of $G$ is nearly gem-free (see Figure~\ref{HC} for the graph $H^*$).
\end{theorem}

\begin{proof}
Let $G'$ be an atom of $G$.  We want to show that $G'$ is nearly
gem-free, so let us assume on the contrary that there is a vertex $v
\in V(G')$ such that the anti-neighborhood of $v$ in $G'$ contains an
induced gem $H$.  Let $H$ have vertex set $\{v_1, v_2, v_3, v_4,
v_5\}$ and edge set $\{v_1v_2, v_2v_3, v_3v_4, v_1v_5, v_2v_5,$
$v_3v_5, v_4v_5\}$.  For $i\in \{1, 2,  \ldots, 5\}$, we define sets $A_i$,
$A_i^+$, $A_i^-$, and $Q$, with respect to $G$, $v$ and $H$, as in the last paragraph of Section~\ref{NT}.  Then
we have the following properties:

\begin{clm}\label{clm:neigT2}
Every vertex $x$ in $N(H)$ satisfies either
$v_5\in N(x)$ or $x$ has at least one neighbor in $\{v_2, v_3\}$.  In particular, (i) if $x
\in A_1$, then $N_H(x) =\{v_5\}$, and (ii)  if $x \in A_2$, then $N_A(x) \in \{ \{v_2, v_3\}, \{v_2, v_5\}, \{v_3, v_5\},  \{v_1, v_3\},$ $ \{v_2, v_4\}\}$.
\end{clm}

\no{\it Proof of Claim~\ref{clm:neigT2}}. Suppose to the contrary that $xv_5\notin
E$ and $x$ has no  neighbor in $\{v_2, v_3\}$.  Then, up to symmetry, we have $xv_1\in
E$, and so $\{x, v_1, v_2, v_3, v_5\}$ induces a co-chair in $G$, a
contradiction.    $\Diamond$

\medskip
Let $B^*$ denote the set $\{x \in A_2 : N_A(x) = \{v_1, v_3\}\} \cup \{x \in A_2 : N_A(x) = \{v_2, v_4\}\}$.

\begin{clm}\label{clm:T2-A2+A3+A4+-empty}
$A_2^+ = A_3^+ = A_4^+ = \emptyset$.
\end{clm}

\no{\it Proof of Claim~\ref{clm:T2-A2+A3+A4+-empty}}.

Assume the contrary and let $x \in
A_2^+\cup A_3^+\cup A_4^+$.  There is a vertex $z$ in $Q$ such that
$xz \in E$.  First suppose that $xv_5 \notin E$. Then by  Claim~\ref{clm:neigT2}, $x$ has at least one neighbor in $\{v_2, v_3\}$. Now, if $\{v_2, v_3\} \subseteq N_H(x)$, then $\{v_5, v_2, v_3, x, z\}$ induces a co-chair in $G$, which is a contradiction. So, we may assume that $x$ has exactly one neighbor in $\{v_2, v_3\}$, say $xv_2 \in E$ and $xv_3 \notin E$. Since $x\in A_i^+$ ($i \ge 2$), either $xv_1 \in E$ or $xv_4 \in E$. But, then either $\{v_5, v_1, v_2, x, z\}$ induces a co-chair in $G$, or $\{v_5, v_2, v_3, v_4, x ,z\}$ induces a $H^*$ in $G$, respectively, a contradiction.
So, suppose that $xv_5 \in E$. Then it follows that there is a clique
$\{p,q,r\}\subset V(H)$ such that $xp, xq\in E$ and $xr\notin
E$.  Then $\{z, x, p, q, r\}$ induces a co-chair in $G$, a
contradiction.  $\Diamond$

\begin{clm}\label{clm:T2-A5+-Ai-complete}
For each $i\in\{2,\ldots,5\}$, $A_5^+$ is
complete to $A^-_i$.
\end{clm}

\no{\it Proof of Claim~\ref{clm:T2-A5+-Ai-complete}}.
Assume the contrary.  Let $x \in A_5^+$
and $y \in A^-_i$ be such that $xy \notin E$.  Since $x \in A_5^+$,
there exists $z \in Q$ such that $xz \in E$.  Now, if $y \in (\bigcup_{i =2}^5 A^-_i)\setminus B^*$, then there exist vertices
$p,q \in V(H)$ such that $pq \in E$ and $yp, yq\in E$.  Then $\{z, x, p, q, y\}$ induces a co-chair in $G$, a
contradiction.  So, $y \in B^*$. But, now,  $\{x, v_4, v_5, v_1, y\}$ induces a co-chair in $G$, a contradiction. $\Diamond$

\begin{clm}\label{clm:T2:A5+clique}
$A_5^+$ is a clique.
\end{clm}

\no{\it Proof of Claim~\ref{clm:T2:A5+clique}}:
Suppose the contrary.  Then $G[A_5^+]$ has
a co-connected component $X$ of size at least~$2$.  Since $G$ is prime, $X$ is not a module in $G$ in $G$, so there is a vertex $z$ in $V(G)\setminus X$ that distinguishes two vertices $x$ and $y$ of
$X$, and since $X$ is co-connected we can choose $x$ and $y$
non-adjacent.  Clearly $z\notin H$ and $z\notin A_5^+$.  So either (i)
$z$ has no neighbor in $H$, or (ii) $z \in A^-$ and so, by Claim~\ref{clm:T2-A5+-Ai-complete}
(since $\{z\}$ is not complete to $A_5^+$), $z \in A^-_1$, or (iii) $z \in
A^+$ and so, by Claim~\ref{clm:T2-A2+A3+A4+-empty}, $z \in A^+_1$.  In either of these three
cases, by Claim~\ref{clm:neigT2}, we see that $\{z, x, y, v_1, v_2\}$ induces a
co-chair in $G$, a contradiction.  $\Diamond$

\medskip

Let $B = A_2^- \cup A_3^- \cup A_4^- \cup A_5^-$.

\begin{clm}\label{clm:T2-A1+-v5-B}
If $A_1^+\neq\emptyset$, then $\{v_5\}$ is
complete to $B$.
\end{clm}
\no{\it Proof of Claim~\ref{clm:T2-A1+-v5-B}}: Assume on the contrary that there is a
vertex $x \in B$ such that $xv_5\notin E(G)$.  Since $A_1^+ \neq \emptyset$, there is a
vertex $a \in A_1^+$ and a vertex $z \in Q$ such that $az \in E(G)$.
Recall that $N_H(a)=\{v_5\}$, by Claim~\ref{clm:neigT2}. Now, if there exists vertices $p, q \in \{v_1, v_2, v_3, v_4\}$ such that $pq \in E$ and
 $xp, xq \in E$, then since $\{x, p, q, v_5, a\}$ does not induce a co-chair in $G$, $xa \in E$. But, then $\{z, a, x, v_5, p, q\}$ induces a $H^*$ in $G$, a contradiction. So, we may assume that $N_H(x) \cap \{v_1, v_2, v_3, v_4\}$ is an independent set. Hence by Claim~\ref{clm:neigT2}, $N_H(x)$ is either $\{v_1, v_3\}$ or $\{v_2, v_4\}$.  We may assume, up to symmetry, that $N_H(x) = \{v_1, v_3\}$. Then since $\{z, a, v_5, v_1, x, v_4\}$ does not induce an $S_{1, 2, 2}$ in $G$, $xa \in E$. But, then $\{z, a, x, v_3, v_2, v_4\}$ induces an $S_{1, 1, 3}$ in $G$, a contradiction.
  $\Diamond$

\begin{clm}\label{clm:T2-A1+-B}
There is no edge between $A_1^+$ and $B$.
\end{clm}

\no{\it Proof of Claim~\ref{clm:T2-A1+-B}}: Assume the contrary, and let $a \in A_1^+$
and $b \in B$ be such that $ab \in E$.  Since $a \in A_1^+$, there
exists $y \in Q$ such that $ay \in E$.  Since $b \in B$, by Claim~\ref{clm:T2-A1+-v5-B}
there exists an index $j$ ($j \in \{1, \ldots, 4\}$) such that $bv_5, bv_j \in
E$.  Then $\{y, a, b, v_5, v_j\}$ induces a co-chair in $G$, which is a
contradiction.  $\Diamond$

\begin{clm}\label{clm:T2-path}
If $a \in A_1^+$, $b \in B$, and $x \in A_1^-$, then $\{a, b, x\}$ does not induce a path in $G$.
\end{clm}

\no{\it Proof of Claim~\ref{clm:T2-path}}: Assume the contrary. Since $a \in A_1^+$, there exists a vertex $z \in Q$ such that $az \in E$.
By Claim ~\ref{clm:T2-A1+-B}, $ab \notin E$. Thus, by the assumption, $ax \in E$ and $xb \in E$. Also, by Claims  \ref{clm:neigT2} and \ref{clm:T2-A1+-v5-B}, we have $av_5, xv_5 \in E$ and $bv_5 \in E$.
 But, now $\{z, a, x, b, v_5\}$ induces a co-chair in $G$, which is a contradiction.
  $\Diamond$

\vspace{0.2cm}

Suppose that $A_1^+ = \emptyset$.  Then $A^+ = A_5^+$, which is a
clique by Claim~\ref{clm:T2:A5+clique}.  Since $A^+$ is a separator in $G$ between $H$ and
$Q$, it follows that $V(G')\cap A^+$ is a clique separator in $G'$
between $H$ and $V(G')\cap Q$ (which contains $v$); this is a
contradiction to the fact that $G'$ is an atom.  Therefore $A_1^+ \neq
\emptyset$.  Now, Claims~\ref{clm:neigT2} and~\ref{clm:T2-A2+A3+A4+-empty} imply that $N(v_4) \subseteq \{v_3,
v_5\}\cup B\cup A_5^+$, and  Claims~\ref{clm:T2:A5+clique}, \ref{clm:T2-A1+-B}, and~\ref{clm:T2-path} imply that $A_5^+ \cup \{v_2,
v_3, v_5\}$ is a clique separator between $\{v_4\}$ and $Q$ in $G$.
Hence $V(G')\cap (A_5^+ \cup \{v_2, v_3, v_5\})$ is a clique separator
between $\{v_4\}$ and $V(G')\cap Q$ in $G'$, again a contradiction to
the fact that $G'$ is an atom.
\end{proof}

\begin{theorem} \label{mwis-SSCH*-time}
The MWIS problem can be solved in polynomial time for ($S_{1, 2, 2}$, $S_{1, 1, 3}$,  co-chair, $H^*$)-free graphs.
\end{theorem}

\begin{proof} Let $G$ be an ($S_{1, 2, 2}$, $S_{1, 1, 3}$,  co-chair, $H^*$)-free
graph.  First suppose that $G$ is prime.  By
Theorem~\ref{SSCH*-free-implies-Gem-free}, every atom of $G$
is nearly gem-free. Since the MWIS in ($S_{1, 2, 2}$, $S_{1, 1, 3}$, co-chair, gem)-free
graphs can be solved in polynomial time (by Theorem~\ref{mwis-SSCG-free-time}),  MWIS in ($S_{1, 2, 2}$, $S_{1, 1, 3}$, co-chair,  $H^*$)-free
graphs can be solved in polynomial time, by Theorem~\ref{thm:Tar}.  Then the time
complexity is the same when $G$ is not prime, by Theorem
\ref{thm:LM}.
\end{proof}

\subsection{($S_{1, 2, 2}$, $S_{1, 1, 3}$, co-chair)-free graphs}

\begin{theorem}\label{thm:SSC-H*-free}
Let $G= (V, E)$ be a prime ($S_{1, 2, 2}$, $S_{1, 1, 3}$, co-chair)-free graph.  Then every
atom of $G$ is nearly $H^*$-free.
\end{theorem}

\begin{proof}
Let $G'$ be an atom of $G$.  We want to show that $G'$ is nearly
$H^*$-free, so let us assume on the contrary that there is a vertex $v
\in V(G')$ such that the anti-neighborhood of $v$ in $G'$ contains an
induced $H^*$ as shown in Figure~\ref{HC}.   For $i=1, \ldots, 6$ we define sets $A_i$,
$A_i^+$, $A_i^-$, and $Q$, with respect to $G$, $v$ and $H$, as in the last paragraph of Section~\ref{NT}.  Then
we have the following properties:

\medskip

\begin{clm}\label{clm:A1}
 $A_1 = \emptyset$.
\end{clm}

\no{\it Proof of Claim~\ref{clm:A1}}. Suppose to the contrary that $A_1 \neq \emptyset$, and let $x \in A_1$.
Then:  (i) If $N_{H^*}(x)$ is either $\{v_1\}$ or $\{v_3\}$, then $\{x\} \cup V(H^*)$ induces a graph which is isomorphic to $H_7$ in $G$, a contradiction to Lemma~\ref{primeSSCfree-forb}. (ii) If $N_{H^*}(x)$ is either $\{v_2\}$ or $\{v_4\}$, then $\{x, v_1, v_2, v_3, v_4\}$ induces a co-chair in $G$, which is a contradiction. (iii) If $N_{H^*}(x) = \{v_5\}$, then  $\{x\} \cup V(H^*)$ induces a graph which is isomorphic to $H_6$ in $G$, a contradiction to Lemma~\ref{primeSSCfree-forb}. (iv) If $N_{H^*}(x) = \{v_6\}$, then  $\{x\} \cup V(H^*)$ induces a graph which is isomorphic to $H_4$ in $G$, a contradiction to Lemma~\ref{primeSSCfree-forb}.
So, the claim holds. $\Diamond$

\medskip
\begin{clm}\label{clm:A2}
 If $x \in A_2$, then $N_{H^*}(x) \in \{\{v_1, v_2\}, \{v_1, v_4\}, \{v_1, v_5\}, \{v_1, v_6\}, $ $\{v_2, v_3\}, \{v_3, v_4\}, $ $\{v_3, v_5\}, \{v_3, v_6\}, \{v_5, v_6\}\}$.
\end{clm}

\no{\it Proof of Claim~\ref{clm:A2}}. For, otherwise if $N_{H^*}(x) \in \{\{v_1, v_3\}, \{v_2, v_5\}, \{v_2, v_6\},$ $ \{v_4, v_5\}, \{v_4, v_6\}\}$, then $\{x, v_1, v_2, v_3, v_4\}$ induces a co-chair in $G$, which is a contradiction, and if $N_{H^*}(x)$ is $\{v_2, v_4\}$, then $\{x\} \cup V(H^*)$ induces a graph which is isomorphic to $H_5$ in $G$, a contradiction to Lemma~\ref{primeSSCfree-forb}. So, the claim holds. $\Diamond$

\medskip

\begin{clm}\label{clm:A2+}
    $A_2^+ = \emptyset$.
\end{clm}

\no{\it Proof of Claim~\ref{clm:A2+}}. Suppose to the contrary that $A_2^+ \neq \emptyset$, and let $x \in A_2^+$.  Then there is a vertex $z$ in $Q$ such that
$xz \in E$. We use  Claim~\ref{clm:A2} to derive a contradiction to our assumption as follows: Now, if $N_{H^*}(x) \in \{\{v_1, v_2\}, \{v_1, v_4\}, \{v_2, v_3\}, \{v_3, v_4\}\}$, then  it follows that there is a clique $\{p,q,r\}\subset V(H^*)$ such that $xp, xq \in E$ and $xr\notin
E$.  But, then $\{z, x, p, q, r\}$ induces a co-chair in $G$, which is a contradiction. Next, if $N_{H^*}(x) \in \{\{v_1, v_5\},$ $ \{v_1, v_6\}, \{v_3, v_5\}, \{v_3, v_6\}\}$, then $\{z, x, v_1, v_3, v_5, v_6\}$ induces an $S_{1, 2, 2}$ in $G$, which is a contradiction. Finally, if  $N_{H^*}(x)$ is $\{v_5, v_6\}$, then  $\{v_1, v_2, v_3,$ $ v_4, v_5, x, z\}$ induces a graph which is isomorphic to $H_4$ in $G$, a contradiction to Lemma~\ref{primeSSCfree-forb}. So, $A_2^+ = \emptyset$, and the claim holds. $\Diamond$

\medskip
\begin{clm}\label{clm:A3}
 If $x \in A_3$, then $N_{H^*}(x) \in \{\{v_1, v_2, v_4\}, \{v_1, v_3, v_5\}, \{v_1, v_5, v_6\}, $ $\{v_2, v_3, v_4\},  \{v_2, v_4,$ $ v_6\},$ $ \{v_3, v_5, v_6\}\}$.
\end{clm}

\no{\it Proof of Claim~\ref{clm:A3}}. Suppose the contrary. Now, if $v_6 \in N_{H^*}(x)$, then since $x \in A_3$, $|N_{H^*}(x) \cap \{v_1, v_2, v_3, v_4\}| \in \{1, 2\}$. If $|N_{H^*}(x) \cap \{v_1, v_2, v_3, v_4\}| = 2$, then it follows that there is a clique $\{p, q, r\}\subset \{v_1, v_2, v_3, v_4\}$ such that $xp, xq \in E$ and $xr\notin E$.  But, then $\{v_6, x, p, q, r\}$ induces a co-chair in $G$, which is a contradiction. So, $|N_{H^*}(x) \cap \{v_1, v_2, v_3, v_4\}| = 1$, and hence $v_5 \in N_{H^*}(x)$. Since $x \in A_3$ and by our contrary assumption, either $v_2 \in N_{H^*}(x)$ or $v_4 \in N_{H^*}(x)$. But, then $\{v_1, v_2, v_3, v_4, x\}$ induces a co-chair in $G$, which is a contradiction. So, we may assume that $v_6 \notin N_{H^*}(x)$. Now, (i) if $N_{H^*}(x)$ is $\{v_1, v_2, v_3\}$, then $\{x, v_1, v_3, v_4, v_5\}$ induces a co-chair in $G$, (ii) if $N_{H^*}(x)$ is $\{v_2, v_3, v_5\}$, then $\{x, v_2, v_3, v_5, v_6\}$ induces a co-chair in $G$, and (iii) if $N_{H^*}(x)$ is $\{v_1, v_2, v_5\}$, then $\{x, v_1, v_2, v_5, v_6\}$ induces a co-chair in $G$, which are contradictions. Finally, if $N_{H^*}(x)$ is $\{v_2, v_4, v_5\}$, then $\{x\} \cup V(H^*)$ induces a graph which is isomorphic to $H_1$ in $G$, a contradiction to Lemma~\ref{primeSSCfree-forb}.
Hence the claim is proved.  $\Diamond$

\medskip
\begin{clm}\label{clm:A3+}
 If $x \in A_3^+$, then $N_{H^*}(x)$ is either $\{v_1, v_5, v_6\}$ or $\{v_3, v_5, v_6\}$.
\end{clm}

\no{\it Proof of Claim~\ref{clm:A3+}}.  For, otherwise, by Claim~\ref{clm:A3},  $N_{H^*}(x) \in \{\{v_1, v_2, v_4\},$ $ \{v_1, v_3, v_5\}, \{v_2, v_3,$ $ v_4\},  \{v_2, v_4, v_6\}\}$. Since $x\in A_3^+$, there is a vertex $z$ in $Q$ such that
$xz \in E$. Now, if $N_{H^*}(x) \in \{\{v_1, v_2, v_4\},$ $ \{v_1, v_3, v_5\}, \{v_2, v_3, v_4\}\}$, then it follows that there is a clique $\{p,q,r\}\subset V(H^*)$ such that $xp, xq \in E$ and $xr\notin E$.  But, then $\{z, x, p, q, r\}$ induces a co-chair in $G$, which is a contradiction. Next, if $N_{H^*}(x)$ is $\{v_2, v_4, v_6\}$, then $\{v_1, v_2, v_3, v_4, v_6, x, z\}$ induces a graph which is isomorphic to $H_6$ in $G$, a contradiction to Lemma~\ref{primeSSCfree-forb}.
So the claim is proved.  $\Diamond$

\medskip
Let $B_3'$ denotes the set $\{x \in A_3^+ \mid  N_{H^*}(x) = \{v_1, v_5, v_6\}\}$ and let $B_3''$ denotes the set $\{x \in A_3^+ \mid N_{H^*}(x) = \{v_3, v_5, v_6\}\}$.

\medskip
\begin{clm}\label{clm:B3clique}
$B_3'$ and $B_3''$ are cliques in $G$.
\end{clm}

\no{\it Proof of Claim~\ref{clm:B3clique}}. Suppose to the contrary that there exists vertices $x, y \in B_3'$ such that $xy \notin E$. Since $x\in A_3^+$, there exists a  vertex $z$ in $Q$ such that $xz \in E$. Now, if $yz \in E$, then $\{z, x, y, v_5, v_6, v_1, v_3\}$ induces a graph which is isomorphic to $H_5$ in $G$, which contradicts Lemma~\ref{primeSSCfree-forb}, and if $yz \notin E$, then $\{v_5, v_6, x, y, z\}$ induces a co-chair in $G$, which is a contradiction. Hence, $B_3'$ is a clique in $G$. Similarly, $B_3''$ is also a clique in $G$. $\Diamond$

\medskip
\begin{clm}\label{clm:B3empty}
At most one of  $B_3'$ or $B_3''$ is non-empty.
\end{clm}

\no{\it Proof of Claim~\ref{clm:B3empty}}. Suppose the contrary, and let $x \in B_3'$ and $y \in B_3''$. Then since $\{x, y, v_1, v_5, v_6\}$ does not induce a co-chair in $G$, $xy \in E$.    Since $x\in A_3^+$, there exists a  vertex $z$ in $Q$ such that $xz \in E$. Now, if $yz \in E$, then $\{v_2, v_5, x, y, z\}$ induces a co-chair in $G$, which is a contradiction, and if $yz \notin E$, then $\{z, x, y, v_3, v_2, v_4\}$ induces an $S_{1, 1, 3}$ in $G$, which is a contradiction. Hence the claim. $\Diamond$

\medskip
\begin{clm}\label{clm:A4+}
$A_4^+ = \emptyset$.
\end{clm}

\no{\it Proof of Claim~\ref{clm:A4+}}.  Suppose to the contrary that $A_4^+ \neq \emptyset$ and let $x \in A_4^+$. There is a vertex $z$ in $Q$ such that $xz \in E$. Now, if $x$ is adjacent to all the vertices in $\{v_1, v_2, v_3, v_4\}$, then $\{z, x, v_1, v_2, v_4, v_4, v_6\}$ induces a graph which is isomorphic to $H_7$ in $G$, a contradiction to Lemma~\ref{primeSSCfree-forb}. So, we may assume that $x$ is non-adjacent to at least one vertex in
 $\{v_1, v_2, v_3, v_4\}$. Also, since $x \in A_4^+$, $x$ is adjacent to at least two vertices in $\{v_1, v_2, $ $v_3, v_4\}$. Now, if  $N_{H^*}(x)$ is $\{v_2, v_4, v_5, v_6\}$, then $\{x, v_1, v_2, v_5, v_6\}$ induces a co-chair in $G$, a contradiction, and in all the other cases, there is a clique $\{p,q,r\}\subset V(H^*)$ such that $xp, xq \in E$ and $xr\notin E$.  But, then $\{z, x, p, q, r\}$ induces a co-chair in $G$, which is a contradiction.
 $\Diamond$

 \medskip
\begin{clm}\label{clm:A5+}
$A_5^+ = \emptyset$.
\end{clm}

\no{\it Proof of Claim~\ref{clm:A5+}}.  Suppose to the contrary that $A_5^+ \neq \emptyset$ and let $x \in A_5^+$.  Then there is a vertex $z$ in $Q$ such that $xz \in E$. Suppose $x$ is adjacent to all the vertices in $\{v_1, v_2, v_3, v_4\}$.
Further, if $x$ is adjacent to $v_5$, then $\{x, v_2, v_3, v_5, v_6\}$ induces a co-chair in $G$, which is a contradiction,  and if $x$ is adjacent to $v_6$, then $\{x\} \cup V(H^*)$ induces a graph which is isomorphic to $H_2$ in $G$, a contradiction to Lemma~\ref{primeSSCfree-forb}. So, we may assume that $x$ is non-adjacent to exactly one vertex in $\{v_1, v_2, v_3, v_4\}$. Then, there is a clique $\{p,q,r\}\subset V(H^*)$ such that $xp, xq \in E$ and $xr\notin E$.  But, then $\{z, x, p, q, r\}$ induces a co-chair in $G$, which is a contradiction.  $\Diamond$

\medskip
\begin{clm}\label{clm:A3+-A6+-clique}
$A_3^+$ is complete to $A_6^+$.
\end{clm}

\no{\it Proof of Claim~\ref{clm:A3+-A6+-clique}}.  Suppose to the contrary that there exist vertices $x \in A_3^+$ and $y \in A_6^+$ such that $xy \notin E$.
Then by Claim~\ref{clm:A3+}, $N_{H^*}(x)$ is either $\{v_1, v_5, v_6\}$ or $\{v_3, v_5, v_6\}$. Then either $\{v_3, y, v_5, v_6, x\}$ or $\{v_1, y, v_5, v_6, x\}$ induces a co-chair in $G$, a contradiction.  $\Diamond$

\medskip
\begin{clm}\label{clm:A6+-A--clique}
 For each $i\in\{2,\ldots, 6\}$, $A_6^+$ is
complete to $A^-_i$.
\end{clm}

\no{\it Proof of Claim~\ref{clm:A6+-A--clique}}.  Assume the contrary.  Let $x \in A_6^+$
and $y \in A^-_i$ be such that $xy \notin E$.  Since $x \in A_6^+$,
there exists $z \in Q$ such that $xz \in E$.  Now, if there exists vertices $p, q \in
V(H^*)$ such that $pq \in E$ and $py, qy \in E$, then $\{z, x, p,
q, y\}$ induces a co-chair in $G$, which is a contradiction. So, we assume that $N_{H^*}(y)$ is an independent set. Then by the above claims on $A_i^+$, $i \ge 2$, we have $N_{H^*}(y) \in \{\{v_1, v_5\}, \{v_1, v_6\}, \{v_3, v_5\}, \{v_3, v_6\},  \{v_2, v_4, v_6\}\}$. Now, if $N_{H^*}(y) \in \{\{v_1, v_5\}, \{v_3, v_5\}\}$, then $\{z, x, y\} \cup V(H^*)$ induces a graph which is isomorphic to $H_8$ in $G$, a contradiction to Lemma~\ref{primeSSCfree-forb}. So, $N_{H^*}(y) \in \{\{v_1, v_6\}, \{v_3, v_6\},  \{v_2, v_4, v_6\}\}$. But, then $\{v_1, v_4, v_5, x, y\}$ induces a co-chair in $G$ (if $N_{H^*}(y) = \{v_1, v_6\})$, and $\{v_3, v_4, v_5, x, y\}$ induces a co-chair in $G$ (if $N_{H^*}(y) = \{v_3, v_6\})$, which are contradictions. Finally, if  $N_{H^*}(y) = \{v_2, v_4, v_6\}$, then $\{z, x, y\} \cup V(H^*)$ induces a graph which is isomorphic to $H_3$ in $G$, a contradiction to Lemma~\ref{primeSSCfree-forb}.
  $\Diamond$

\medskip
\begin{clm}\label{clm:A6+-clique}
$A_6^+$ is a clique.
\end{clm}

\no{\it Proof of Claim~\ref{clm:A6+-clique}}: Suppose the contrary.  Then $G[A_6^+]$ has
a co-connected component $X$ of size at least~$2$.  Since $G$ is
prime, $X$ is not a non-trivial module in $G$, so there is a vertex $z$
in $V(G)\setminus X$ that distinguishes two vertices $x$ and $y$ of
$X$, and since $X$ is co-connected we can choose $x$ and $y$
non-adjacent.  We may assume (wlog.) that $xz \in E$  and $yz \notin E$. Clearly $z\notin V(H^*)$ and $z \notin A_6^+$.  By Claims~\ref{clm:A1} and \ref{clm:A6+-A--clique}, $z \notin A^-$, and by Claims~\ref{clm:A1}, \ref{clm:A2+},  \ref{clm:A4+}, \ref{clm:A5+}, and \ref{clm:A3+-A6+-clique}, we have $z \notin A^+$.
Hence, $z$ has no neighbor in $H^*$, and we see that $\{z, x, y, v_1, v_2\}$ induces a
co-chair in $G$, a contradiction.  $\Diamond$

\medskip

By Claims~\ref{clm:A1}, \ref{clm:A2+}, \ref{clm:A4+}, \ref{clm:A5+}, we have $A^+ = A_3^+ \cup A_6^+$, which is a
clique by Claims~\ref{clm:B3clique}, \ref{clm:B3empty}, \ref{clm:A3+-A6+-clique}, and \ref{clm:A6+-clique}.  Since $A^+$ is a separator in $G$ between $H^*$ and
$Q$, it follows that $V(G')\cap A^+$ is a clique separator in $G'$
between $H^*$ and $V(G')\cap Q$ (which contains $v$); this is a
contradiction to the fact that $G'$ is an atom.
\end{proof}

\begin{theorem} \label{mwis-SSC-time}
The MWIS problem can be solved in polynomial time for ($S_{1, 2, 2}$, $S_{1, 1, 3}$, co-chair)-free graphs.
\end{theorem}

\begin{proof} Let $G$ be an ($S_{1, 2, 2}$, $S_{1, 1, 3}$, co-chair)-free
graph.  First suppose that $G$ is prime.  By
Theorem~\ref{thm:SSC-H*-free}, every atom of $G$
is nearly $H^*$-free.  Since the MWIS problem in ($S_{1, 2, 2}$, $S_{1, 1, 3}$, $H^*$, co-chair)-free graphs can be solved in  polynomial time (by Theorem~\ref{mwis-SSCH*-time}), MWIS can be solved in polynomial time for $G$, by Theorem~\ref{thm:Tar}.  Then the time
complexity is the same when $G$ is not prime, by Theorem
\ref{thm:LM}.
\end{proof}

\noindent{\bf Acknowledgement}: The author sincerely thanks Prof.~Vadim V. Lozin and Prof.~Fr\'{e}d\'{e}ric Maffray for the fruitful discussions, and for their valuable suggestions and comments.

\end{document}